\pdfoutput=1

\documentclass[conference]{IEEEtran}
\ifCLASSINFOpdf
\else
\fi

\usepackage[utf8]{inputenc}
\usepackage{amsmath}
\usepackage{amsfonts}

\usepackage{amsthm}

\newenvironment{definition}[1][Definition]{\begin{trivlist}
\item[\hskip \labelsep {\bfseries #1}]}{\end{trivlist}}

\usepackage[table,xcdraw]{xcolor}

\usepackage{graphicx}
\graphicspath{{Figures/}}
\usepackage{epstopdf}

\usepackage{etex}

\usepackage{epsfig}
\usepackage{psfrag}
\usepackage{subfigure}
\usepackage{url}
\usepackage{comment}
\usepackage{afterpage}
\usepackage{cite}
\usepackage{pstricks}
\usepackage{graphics}
\usepackage{soul} 
\usepackage{multirow}

\usepackage{todonotes} 

\usepackage{graphicx, array, blindtext} 

\usepackage{pifont}

\usepackage{lipsum}

\usepackage{graphicx}
\usepackage[font=small,labelfont=bf]{caption}
\usepackage[ruled]{algorithm2e}

\SetAlFnt{\small}
\SetAlCapFnt{\small}
\SetAlCapNameFnt{\small}
\SetAlCapHSkip{0pt}
\IncMargin{-\parindent}

\usepackage{upgreek}
\usepackage{wrapfig}

\usepackage{fixfoot}

\usepackage{listings}
\usepackage{psfrag,wrapfig}
\usepackage{xspace}

\usepackage{enumerate}
\usepackage{paralist} 
\usepackage{subfiles} 

\usepackage{mathtools}
\DeclarePairedDelimiter\ceil{\lceil}{\rceil}
\DeclarePairedDelimiter\floor{\lfloor}{\rfloor}

\newcommand{\eg}{{\it e.g.,}\xspace}

\newcommand{\ie}{{\it i.e.,}\xspace}

\newcommand{\ci}{{\it (i) }}
\newcommand{\cii}{{\it (ii) }}

\newtheorem{theorem}{Theorem}

\newtheorem{remark}{Remark}

\usepackage[nodisplayskipstretch]{setspace}
\usepackage{simplemargins}
\setleftmargin{0.6in}
\setrightmargin{0.6in}
\settopmargin{.6in}
\setbottommargin{0.5in}

\begin{document}
%
\title{
Restart-Based Fault-Tolerance:\\ System Design and Schedulability Analysis}

\author{\IEEEauthorblockN{Fardin Abdi, Renato Mancuso, Rohan Tabish, Marco Caccamo}
\IEEEauthorblockA{Department of Computer Science, 
	University of Illinois at Urbana-Champaign, USA\\
\{abditag2, rmancus2, rtabish, mcaccamo\}@illinois.edu}
}


%


\maketitle
\bstctlcite{IEEEexample:BSTcontrol}

\thispagestyle{plain}
\pagestyle{plain}

\begin{abstract}
	
Embedded systems in safety-critical environments are continuously required to deliver more performance and functionality, while expected to provide verified safety guarantees.
Nonetheless, platform-wide software verification~(required for safety) is often expensive. Therefore, design methods that enable utilization of components such as real-time operating systems~(RTOS), without requiring their correctness to guarantee safety, is necessary. 

In this paper, we propose a design approach to deploy safe-by-design embedded systems. To attain this goal, we rely on a small core of verified software to handle faults in applications and RTOS and recover from them while ensuring that timing constraints of safety-critical tasks are always satisfied. Faults are detected by monitoring the  application timing and fault-recovery is achieved via full platform restart and software reload, enabled by the short restart time of embedded systems. Schedulability analysis is used to ensure that the timing constraints of critical plant control tasks are always satisfied in spite of faults and consequent restarts. We derive schedulability results for four restart-tolerant task models. We use a simulator to evaluate and compare the performance of the considered scheduling models. 
\end{abstract}


%
\IEEEpeerreviewmaketitle

\section{Introduction}

Embedded controllers with smart capabilities are being increasingly
used to implement safety-critical cyber-physical systems (SC-CPS). In
fact, modern medical devices, avionic and automotive systems, to name
a few, are required to deliver increasingly high performance without
trading off in robustness and assurance. Unfortunately, satisfying the
increasing demand for smart capabilities and high performance means
deploying increasingly complex systems. Even seemingly simple embedded
control systems often contain a multitasking real-time kernel, support
networking, utilize open source
libraries~\cite{sulaman2014development}, and a number of specialized
hardware components (GPUs, DSPs, DMAs, \emph{etc.}). As systems
increase in complexity, however, the cost of formally verifying their
correctness can easily explode.




Testing alone is insufficient to guarantee the correctness of
safety-critical systems, and unverified software may violate system safety
in multiple ways, for instance: (i) the control application may
contain unsafe logic that guides the system towards hazardous states;
(ii) the logic may be correct but incorrectly implemented thereby
creating unsafe commands at runtime~(application-level faults); (iii)
even with logically safe, correctly implemented control applications,
faults in underlying software layers (e.g. RTOS and device
drivers) can prevent the correct execution of the controller and
jeopardize system safety~(system-level faults). Due to the limited
feasibility and high cost of platform-wide formal verification, we
take a different approach. Specifically, we propose a
software/hardware co-design methodology to deploy SC-CPS that (i)
provide strong safety guarantees; and (ii) can utilize unverified
software components to implement complex safety-critical
functionalities.

Our approach relies on a key observation: by performing careful
boot-sequence optimization, many embedded platforms and RTOS utilized
in automotive industry, avionics, and manufacturing can be {\bf
  entirely restarted} within a very short period of time. Restarting a
computing system and reloading a fresh image of all the software~(\ie
RTOS, and applications) from a read-only source appears to be an
effective approach to recover from unexpected faults. Thus, we propose
the following: as soon as a fault that disrupts the execution of
critical components is detected, the entire system is restarted. After
a restart, all the safety-critical applications that were impacted by
the restart are re-executed. If restart and re-execution of critical
tasks can be performed \emph{fast enough}, i.e. such that timing
constraints are always met in spite of task re-executions, the
physical system will remain oblivious to and will not be impacted by
the occurrence of faults.

The effectiveness of the proposed restart-based recovery relies on
timely detection of faults to trigger a restart. Since detecting
logical faults in complex control applications can be challenging, we
utilize Simplex Architecture~\cite{sha1998dependable,
  Sha01usingsimplicity, sha1996evolving} to construct control
software.  Under Simplex, each control application is divided into three tasks; safety controller, complex controller, and decision module. And, safety of the system relies solely on
timely execution of the safety controller tasks. From a scheduling
perspective, safety is guaranteed if safety controller tasks have enough CPU
cycles to re-execute and finish before their deadlines in spite of
restarts. In this paper, we analyze the conditions for a periodic task
set to be schedulable in the presence of restarts and
re-executions. We assume that when a restart occurs, the task instance
executing on the CPU and any of the tasks that were preempted before
their completion will need to re-execute after the restart. In
particular, we make the following contributions:

\begin{itemize}
\item We propose a Simplex Architecture that can be recovered via
  restarts and implemented on a {\bf single processing unit};
\item We derive the response time analysis under fixed-priority with
  fully preemptive and fully non-preemptive disciplines in presence of
  restart-based recovery and discuss pros and cons of each one;
\item We propose response time analysis of fixed-priority scheduling
  in presence of restarts for tasks with preemption
  thresholds~\cite{wang1999scheduling} and non-preemptive ending
  intervals~\cite{1508455} to improve feasibility of task sets;
\end{itemize}

\section{Background on Simplex Design}

Our proposed approach is designed for the control tasks that are constructed following Simplex verified design guidelines~\cite{sha1998dependable, Sha01usingsimplicity, sha1996evolving}. In the following, we review Simplex design concepts which are essential for the understanding the methodology of this paper. The goal of original Simplex approach is to design controllers, such that the faults in controller software do not cause the physical plant to violate its safety conditions.

\begin{definition}
States of the physical plant that do not violate any of the safety conditions are referred to as \textit{admissible states}.
 The physical subsystem is assumed safe as long it is in an admissible state. Likewise those that violate the constraints are referred to as \textit{inadmissible states}.
\end{definition}

Under Simplex Architecture, each controlled physical process/component requires a safety controller, a complex controller, and a decision module. In the following, we define properties of each component. 

\begin{definition}
\textit{Safety Controller} is a controller for which a subset of the admissible states called \textit{recoverable states} exists with the following property; If the safety controller starts controlling the plant from one of those states, all future states will remain admissible. The set of recoverable states is denoted by $\mathcal{R}$. Safety controller is formally verified \ie it does not contain logical or implementation errors. 
\end{definition}

\begin{definition}
\textit{Complex Controller} is the main controller task of the system that drives the plant towards mission set points. However, it is unverified \ie it may contain unsafe logic or implementation bugs. As a result, it may generate commands that force the plant into inadmissible states.
\end{definition}

\begin{definition}
\textit{Decision Module} includes a switching logic that can determine if the physical plant will remain safe~(stay within the admissible states) if the control output of complex controller is applied to it.
\end{definition}

There are multiple approaches to design a verified safety controller and decision module. The first proposed way is based on solving linear matrix inequalities~\cite{seto1999case}, which has been used to design Simplex systems as complicated as automated landing maneuvers for an F-16~\cite{seto2000case}. 
According to this approach, safety controller is designed by approximating the system with
linear dynamics in the form: $\dot{x} = Ax + Bu$, for state vector
$x$ and input vector $u$. In this
approach, \emph{safety constraints} are expressed as linear
constraints in the form of linear matrix inequalities. These constraints, along with the linear
dynamics for the system, are the inputs to a convex optimization
problem that produces both linear proportional controller gains $K$,
as well as a positive-definite matrix $P$. The resulting linear-state feedback controller, $u = Kx$, yields closed-loop
dynamics in the form of $\dot{x} = (A + BK)x$. Given a state $x$, when the input $Kx$
is used, the $P$ matrix defines a
Lyapunov potential function $(x^TPx)$ with a negative-definite
derivative. As a result,
the stability of the physical plant is guaranteed using Lyapunov's
direct or indirect methods. Furthermore, matrix $P$ defines an ellipsoid in the state space where all safety constraints are
satisfied when $x^TPx < 1$. If sensors' and actuators' saturation
points were provided as constraints, the states inside
the ellipsoid can be reached using control commands within the sensor/actuator limits.

In this way, when the gains $K$ define the safety controller, the ellipsoid of states $x^T P x < 1$ is the set of recoverable states~$\mathcal{R}$. This ellipsoid is used to determine the proper switching logic of the decision module. As long as the system remains inside the ellipsoid, any unverified, complex controller can be used. If the state approaches the boundary of the ellipsoid, control can be switched to the safety controller which will drive the system towards the equilibrium point where $x^T P x = 0$.

An alternative approach for constructing a verified safety controller and decision module is proposed in~\cite{rt-reach}. Here, safety controller is constructed similar to the above approach~\cite{seto1999case}. 
However, a novel switching logic is proposed for decision module to decide about the safety of complex controller commands. Intuitively, this check is examining what happens if the complex controller is used for a single control interval of time, and then the safety controller is used thereafter. If the reachable states contain an inadmissible state (either before the switch or after), then the complex controller cannot be used for one more control interval. Assuming the system starts in a recoverable state, this guarantees it will remain in the recoverable set for all time.


A system that adheres to this architecture is guaranteed to remain safe only if safety controller and decision module execute correctly. In this way, the safety premise is valid only if safety controller and decision module execute in every control cycle. Original Simplex design, only protects the plant from faults in the complex controller. For instance, if a fault in the RTOS crashes the safety controller or decision module, safety of the physical plant will get violated.

\section{System Model and Assumptions}
\label{sec:methodology}
In this section we formalize the considered system and task
model, and discuss the assumptions under which our methodology is
applicable.

\subsection{Periodic Tasks}

We consider a task set $\mathcal{T}$ composed of $n$ periodic tasks
$\tau_1 \ldots \tau_n$ executed on a uniprocessor under fixed priority
scheduling. Each task $\tau_i$ is assigned a priority level
$\pi_i$. We will implicitly index tasks in decreasing priority order,
\ie , $\tau_i$ has higher priority than $\tau_k$ if $i < k$.  Each
\emph{periodic task} $\tau_i$ is expressed as a tuple $(C_i, T_i,
D_i, \phi_i)$, where $C_i$ is the worst-case execution time~(WCET), 
$T_i$ is the period, $D_i$ is the relative deadline of each
task instance, and $\phi_i$ is the phase~(the release time
of the first instance). The following relation holds: $C_i \le D_i \le
T_i$. Whenever $D_i = T_i$ and $\phi_i =0$, we simply express tasks parameters as
$(C_i, T_i)$. Each instance of a periodic task is called \emph{job}
and $\tau_{i,k}$ denotes the $k$-th job of task $\tau_i$. Finally,
$hp(\pi_i)$ and $lp(\pi_i)$ refer to the set of tasks with higher or
lower priority than $\pi_i$ \ie $hp(\pi_i) = \{\tau_j\mid
\pi_i<\pi_j\}$ and $lp(\pi_i) = \{\tau_j\mid \pi_i > \pi_j\}$.  We
indicate with $T_r$ the minimum inter-arrival time of faults and
consequent restarts; while $C_r$ refers to the time required to
restart the system.



\subsection{Critical and Non-Critical Workload}

It is common practice to execute multiple controllers for different
processes of physical plant on a single processing unit. In this work,
we use the Simplex Architecture~\cite{sha1998dependable, Sha01usingsimplicity, sha1996evolving} to implement each
controller. As a result, three periodic tasks are associated with
every controller: (i) a safety controller~(SC) task, (ii) a complex
controller~(CC) task, and (iii) a decision module~(DM) task. In
typical designs, the three tasks that compose the same controller have
the same period, deadline, and release time.

\begin{remark}
  SC's control command is sent to the actuator buffer immediately
  before the termination of that job instance. Hence, the timely
  execution of SC tasks is {\bf necessary and sufficient} for the
  safety of the physical plant.
  \label{remark:safety}
\end{remark}

As a result, out of the three tasks, SC must execute first and write
its output to the actuator command buffer. Conversely, DM needs to
execute last, after the output of CC is available, to decide if it is
safe to replace SC's command which is already in the actuator
buffer. Hence, the priorities of the controller tasks need to be in
the following order\footnote{We assume enough priority levels to
  assign distinct priorities.}: $\pi(DM) < \pi(CC) < \pi(SC)$.  Note
that, the precedence constraint that SC, CC and DM tasks must execute
in this order can be enforced through the proposed priority ordering
if self-suspension and blocking on resources are excluded and if the
scheduler is work-conserving. We consider fixed priority scheduling,
which is work-conserving and we assume SC, CC and DM tasks do not
self-suspend. Moreover, tasks controlling different components are
independent; SC, CC and DM tasks for the same component share sensors
and actuator channels. Sensors are read-only resources, do not require
locking/synchronization and therefore cannot cause blocking. A given
SC task, may only share actuator channels with the corresponding DM
task. However, SC jobs execute before DM jobs and do not self-suspend,
hence DM cannot acquire a resource before SC has finished its
execution.

The set of all the SC tasks on the system is called \emph{critical}
workload. All the CC and DM tasks are referred as \emph{non-critical}
workload. Safety is guaranteed if and only if all the critical tasks
complete before their deadlines. Whereas, execution of non-critical
tasks is not crucial for safety; these tasks are said to be
mission-critical but not safety-critical. We assume that the first
$n_c$ tasks of $\mathcal{T}$ are critical. Notice that with this
indexing strategy, any critical task has a higher priority than any
non-critical task.




\subsection{Fault Model}

In this paper, we consider two types of fault for the system;
application-level faults and system-level faults. We make the
following assumptions about the faults that our system safely handles:

\begin{description}

\item[A1] The original image of the system software is stored on a
  read-only memory unit (\eg, E$^2$PROM).  This content is
  unmodifiable at runtime.
\item[A2] Application faults may only occur in the unverified
  workload~(\ie all the application-level processes on the system
  except SC and DM tasks).	
\item[A3] SC and DM tasks are independently verified and
  fault-free. They might, however, fail silently~(no output is
  generated) due to faults in software layers or other applications on
  which they depend.
\item[A4] We only consider system- and application-level faults that
  cause SC and DM tasks to fail silently but do not change their logic
  or alter their output.
\item[A5] Faults do not alter sensor readings.
\item[A6] Once SC or CC tasks have send their outputs to the
  actuators, the output is unaffected by system restart. As such, a
  task does not need to be re-executed if it has completed correctly
  before a restart.
\item[A7] Re-executing a task even if it has completed correctly does
  not negatively impact system safety.

\item[A8] Monitoring and initializer tasks~(Section~\ref{sec:detection}) are independently verified and
fault-free. We assume that system faults can only cause silent failures in these tasks~(no output or correct output).

\item[A9] $T_r$ is larger than the least common
  multiple~(hyper-period\footnote{Length of the hyper-period can be
    significantly reduced if the control tasks have harmonic
    periods.}) of critical tasks, i.e.  $T_r > \text{LCM}\{T_k \mid k
  \leq n_c\}$

	
	
\end{description}

\subsection{Scheduler State Preservation and Absolute Time}

In order to know what tasks were preempted, executing, or completed
after a restart occurs, it is fundamental to carry a minimum amount of
data across restarts. As such, our architecture requires the existence
of a small block of non-volatile memory (NVM). We also require the
presence of a monotonic clock unit (CLK) as an external device. CLK is
used to derive the absolute time after a system restart.  Since we
assume periodic tasks, the information provided by CLK is enough to
determine the last release time of each task. Whenever a critical task
is completed, the completion timestamp obtained from CLK is written to
NVM, overwriting the previous value for the same task. We assume that
a timestamp update in NVM write can be performed in a transactional
manner.

\subsection{Recovery Model}
\label{sec:recoverymodel}

The recovery action we assume in this paper is to restart the entire
system, reload all the software~(RTOS and applications) from a
read-only storage unit, and re-execute all the jobs that were released
but not completed at the time of restart. The priority of a
re-executing instance is the same as the priority of the original
job. Within $C_r$ time units, the system~(RTOS and applications)
reloads from a read-only image, and re-execution is initiated as
needed. Figure~\ref{fig:sched_preempt} depicts how restart and task
re-execution affect the scheduling of 3 real-time tasks ($\tau_1,
\tau_2, \text{and} \tau_3$). When the restart happens at $t = 10 -
\epsilon$, $\tau_1$ was still running. Moreover, $\tau_2$ and $\tau_3$
were preempted at time $t = 9$ and $t = 8$, respectively. Hence all
the three task will need to be re-executed after the restart.

\begin{figure}
  \centering
  \includegraphics[trim={.3cm 0.2cm 0cm 0},clip,
  width=0.8\linewidth]{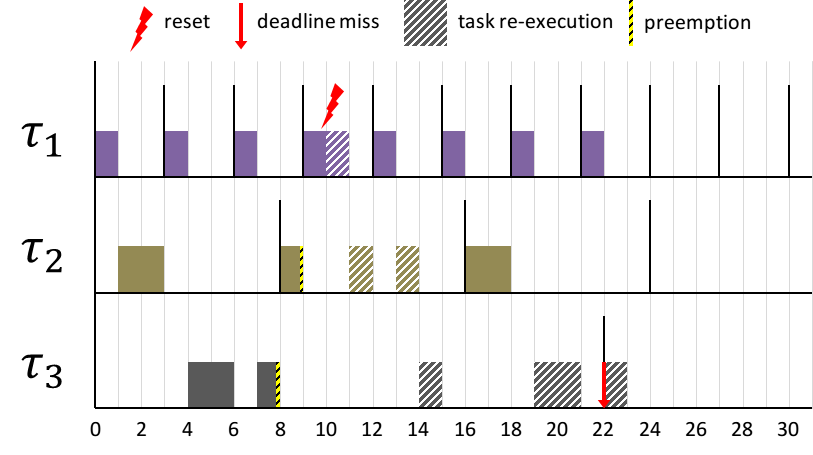}
  \caption{Example of fully preemptive system with 3 tasks $\tau_1 =
    (1, 3); \tau_2 = (2, 8); \tau_3 = (4, 22)$, and restart at $t = 10
    - \epsilon$ ($C_r = 0$). The taskset is schedulable without
    restarts, however, restart and task re-execution causes a deadline
    miss at $t = 22$.}
  \label{fig:sched_preempt}
\end{figure}

System restart is triggered only after a fault is detected. The
following definition of fault is used throughout this paper:

\begin{definition}[Critical Fault:]
  any system misbehavior that leads to a non-timely execution of any
  of the critical tasks.
\end{definition}

It follows that (i) the absence of critical faults guarantees that
every critical task completes on time; that (ii) the timely completion
of all the critical tasks ensures system safety by Assumptions A3-A7;
and that (iii) being able to detect all critical faults and re-execute
critical tasks by their deadline is enough to ensure timely completion
of critical tasks in spite of restarts. We discuss critical fault
detection in Section~\ref{sec:detection}; and we analyze system
schedulability in spite of critical faults in
Section~\ref{sec::schedulability} and~\ref{sec:limited-sched}. Since handling critical faults is necessary and sufficient~(Remark~\ref{remark:safety}) for safety, in the rest of this paper, the term fault is used to refer to critical faults.



\subsection{RBR-Feasibility}
A task set $\mathcal{T}$ is said to be feasible under restart based
recovery~(RBR-Feasible) if the following two conditions are satisfied;
\ci there exists a schedule such that all jobs of all the critical
tasks, or their potential re-executions, can complete successfully
before their respective deadlines, even in the presence of a
system-wide restart, occurring at any arbitrary time during
execution. \cii All jobs, including instances of non-critical tasks,
can complete before their deadlines when no restart is performed.

%

\section{Fault Detection and Task Re-execution}
\label{sec:detection}



As described in the previous section, a successful fault-detection
approach must be able to detect any fault before the
deadline of a \emph{critical} task is missed, and to trigger the recovery
procedure. Another key requirement is being able to correctly
re-execute \emph{critical} jobs that were affected by a restart.

{\bf Fault detection with watchdog (WD) timer:} to explain the
detection mechanism, we rely on the concept of \emph{ideal worst-case
  response time}, i.e. the worst-case response time of a task when
there are no restarts~(and no re-executions) in the system. 
We use $\mathcal{\hat{R}}_i$ to denote the ideal worst-case response
time of $\tau_i$.  $\mathcal{\hat{R}}_i$ can be derived using
traditional response-time analysis, or with the analysis proposed in
Section~\ref{sec::schedulability} and~\ref{sec:limited-sched} by imposing all the overhead terms
$O^x_y = 0$.

%

If no faults occur in the system, every instance of $\tau_i$ is
expected to finish its execution within at most $\mathcal{\hat{R}}_i$
time units after its arrival time. This can be checked at runtime with
a monitoring task. Recall that each critical job records its
completion timestamp $t^{comp}_i$ to NVM. The monitoring task checks
the latest timestamp for $\tau_i$ at time instants $kT_i +
\mathcal{\hat{R}}_i$. If $t^{comp}_i < kT_i$ it means that $\tau_i$
has not completed by its ideal worst-case response time. Hence, a
restart needs to be triggered. A single WD can be used to always
ensure a system reset if any of the critical tasks does not complete
by its ideal worst-case response time. The following steps are
performed:
\begin{enumerate}
  
\item Determine the next checkpoint instant $t_{next}$ and checked
  critical task $\tau_i$ as follows:
  \begin{equation}
    t_{next} = \min_{i \leq c_n}\bigg(\lfloor (t - \phi_i)/T_i \rfloor T_i + \phi_i + \mathcal{\hat{R}}_i\bigg).
  \end{equation}
  In other words, $t_{next}$ captures the earliest instant of time
  that corresponds to the elapsing of the ideal worst-case response
  time of some critical task $\tau_i$;

\item Set the WD to restart the system after $t - t_{next} + \epsilon$
  time units;

\item Terminate and set wake-up time at $t_{next}$;

\item At wake-up, check if $\tau_i$ completed correctly: if
  $t^{comp}_i$ obtained from NVM satisfies $t^{comp}_i \ge \lfloor (t
  - \phi_i)/T_i \rfloor T_i + \phi_i$, then acknowledge the WD so that
  it does not trigger a reset. Otherwise, do nothing, causing a
  WD-induced reset after $\epsilon$ time units.

\item Continue from Step 1 above.
\end{enumerate}

Notice that this simple solution utilizes only one WD timer, and
handles all the silent failures. The advantage of using hardware WD
timers is that if any faults in the OS or other applications, prevent
the time monitor task from execution, the WD which is already set,
will expire and restart the system.

To determine which tasks to execute after a restart, we propose the following. Immediately after the
reboot completes, a initializer task calculates the latest release time of
each task $\tau_i$ using $\lfloor (t-\phi_i)/T_i \rfloor T_i + \phi_i$ where
$t$ is the current time retrieved from CLK. Next, it retrieves the last recorded completion time of the task, $t^{comp}_i$, from NVM. If $t^{comp}_i < \lfloor (t-\phi_i)/T_i \rfloor T_i + \phi_i$, then the task needs to be executed, and is added to the list of ready tasks. It is possible that a task completed its execution prior to the restart, but was not able to record the completion time due to the restart. In this case, the task will be executed again which does not impact the safety due to Assumption~A7.

\section{RBR-Feasibility Analysis}
\label{sec::schedulability}

As mentioned in Section~\ref{sec:detection}, re-execution of jobs
impacted by a restart must not cause any other job to miss a
deadline. Also, re-executed jobs need to meet their deadlines as
well. The goal of this section is to present a set of sufficient
conditions to reason about the feasibility of a given task set
$\mathcal{T}$ in presence of restarts~(RBR-feasibility). In
particular, in Sections~\ref{sec:presched} and~\ref{sec:npresched}, we
present a methodology that provides a sufficient condition for exact
RBR-Feasibility analysis of preemptive and non-preemptive task sets.


\textbf{Definition}: \emph{Length of level-$i$ preemption chain} at
time $t$ is defined as sum of the executed portions of all the tasks
that are in the preempted or running state, and have a priority
greater than or equal to $\pi_i$ at $t$. \emph{Longest level-$i$
  preemption chain} is the preemption chain that has the longest
length over all the possible level-$i$ preemption chains.

For instance, consider a fully preemptive task set with four tasks;
$C_1$ = 1, $T_1 = 5$, $C_2$ = 3, $T_2 = 10$, $C_3$ = 2, $T_3 = 12$,
$C_4$ = 4, $T_3 = 15$, and $\pi_4<\pi_3 < \pi_2 < \pi_1$. For this
task set, the longest level-$3$ and level-$4$ preemption chains are 6
and 10, respectively.

\subsection{Fully Preemptive Task Set}
\label{sec:presched}


Under fully preemptive scheme, as soon as a higher priority task is
ready, it preempts any lower priority tasks running on the
processor. To calculate the worst-case response time of task $\tau_i$,
we have to consider the case where the restart incurs the longest
delay on finishing time of the job. For a fully preemptive task set,
this occurs when every task $\tau_k$ for $k \in \{2,\ldots,i\}$ is
preempted immediately prior to its completion by $\tau_{k-1}$ and
system restarts right before the completion of $\tau_1$. In other
words, when tasks $\tau_1$ to $\tau_i$ form the longest level-i
preemption chain. An example of this case is depicted in
Figure~\ref{fig:sched_preempt}. In this case, the restart and
consequent re-execution causes a deadline miss at $t = 22$. The
example uses only integer numbers for task parameters, hence tasks can
be preempted only up to 1 unit of time before their completion. In the
rest of the paper, we discuss our result assuming that tasks' WCETs
are real numbers.

Theorem~\ref{thm:resp_time_fullypreemptive} provides 
RBR-feasibility conditions for a fully preemptive task set
$\mathcal{T}$, under fixed priority scheduling.

\begin{theorem}
\label{thm:resp_time_fullypreemptive} 
A set of preemptive periodic tasks $\mathcal{T}$ is
RBR-Feasible under fixed priority algorithm if the response time
$R_i$ of each task $\tau_i$ satisfies the condition: $\forall \tau_i
\in \mathcal{T}, R_i \le D_i$. $R_i$ is obtained for the smallest
value of $k$ for which we have $R_i^{(k+1)} = R_i^{(k)}$.
\begin{equation}
\label{eq:responsetimefullypreemptive}
R_i^{(k+1)} = C_i + \sum_{\tau_j \in hp(\pi_i)}\ceil[\bigg]{\frac{R_i^{(k)}}{T_j}} C_j + \mathcal{O}_{i}^{p}
\end{equation}
where the restart overhead $\mathcal{O}_{i}^{p}$ on response time is
\begin{equation}
\label{eq:overheadonfp}
\mathcal{O}_{i}^{p} = \left\{
\begin{array}{lr}
C_r + \sum_{\tau_j \in hp(\pi_i) \cup \{ \tau_i \}} C_j & i\leq n_c\\
0& i > n_c
\end{array}\right.
\end{equation}
\end{theorem}

\begin{proof}
  First, note that Equation~\ref{eq:responsetimefullypreemptive}
  without the overhead term $\mathcal{O}_{i}^{p}$, corresponds to the
  classic response time of a task under fully preemptive fixed
  priority scheduling~\cite{liu1973scheduling}. The additional
  overhead term represents the worst-case interference on the task
  instance under analysis introduced by restart time and the
  re-execution of the preempted tasks. We need to show that the
  overhead term can be computed using
  Equation~\ref{eq:overheadonfp}. Consider the scenario in which every
  task $\tau_k$ is preempted by $\tau_{k-1}$ after executing for
  $\delta_i$ time units where $k \in \{2,...,i\}$. And, a restart
  occurs after $\tau_1$ executed for $\delta_1$ time units. Due to the
  restart, all the tasks have to re-execute and the earliest time
  $\tau_i$ can finish its execution is $C_r
  +\delta_i+...+\delta_1+C_i+...+C_1$. Hence, it is obvious that the
  later each preemption or the restart in $\tau_1$ occurs, the more
  delay it creates for $\tau_i$. Once a task has completed, it no
  longer needs to be re-executed. Therefore, the maximum delay of each
  task is felt immediately prior to the task's completion
  instant. Thus, the overhead is maximized when each $\tau_k$ is
  preempted by $\tau_{k-1}$ for $k \in \{2,..,i\}$ and restart occurs
  immediately before the end of $\tau_1$.
\end{proof}

As seen in this section, the worst-case overhead of restart-based
recovery in fully preemptive setting occurs when system restarts at
the end of longest preemption chain. Therefore, to reduce the overhead
of restarting, length of the longest preemption chain must be
reduced. In order to reduce this effect we investigate the non-preemptive setting in the following section.

%
%
%

\subsection{Fully Non-Preemptive Task set}
\label{sec:npresched}

Under this model, jobs are not preempted until their execution
terminates. At every termination point, the scheduler selects the task
with the highest priority amongst all the ready tasks to execute. The
main advantage of non-preemptive task set is that at most one task
instance can be affected by restart at any instant of time.

Authors in~\cite{Davis2007} showed that in non-preemptive scheduling,
the largest response time of a task does not necessarily occur in the
first job after the critical instant. In some cases, the
high-priority jobs activated during the non-preemptive execution of
$\tau_i$’s first instance are pushed ahead to successive jobs, which
then may experience a higher interference. Due to this phenomenon, the
response time analysis for a task cannot be limited to its first job,
activated at the critical instant, as done in preemptive scheduling,
but it must be performed for multiple jobs, until the processor
finishes executing tasks with priority higher than or equal to
$\pi_i$. Hence, the response time of a task needs to be computed
within the longest \emph{Level-$i$ Active Period}, defined as
follows~\cite{4271700, 6164261}.

\textbf{Definition}: The \emph{Level-$i$ Active Period} $L_i$ is an interval
$[a,b)$ such that the amount of processing that still needs to be
performed at time $t$ due to jobs with priority higher than or equal
to $\pi_i$, released strictly before $t$, is positive for all $t \in
(a,b)$ and null in $a$ and $b$. It can be computed using the following
iterative relation:
\begin{equation}
\label{eq:busyperiod}
L_{i}^{(q)} = B_i + C_i + \sum _{j\in hp(\pi_i)}\lceil {L_i^{(q-1)}}/{T_j} \rceil C_j + \mathcal{O}_i^{np}
\end{equation}
Here, $\mathcal{O}_i^{np}$ is the maximum overhead of restart on the
response time of a task. In the following we describe how to calculate
this value. $L_i$ is the smallest value for which $L_i^{(q)} =
L_i^{(q-1)}$. This indicates that the response time of task $\tau_i$
must be computed for all jobs $\tau_{i,k}$ with $k \in [1,K_i]$ where
$K_i = \lceil L_i/T_i \rceil$.

Theorem~\ref{thm:resp_time_nonp} describes the sufficient conditions
under which a fault and the subsequent restart do not compromise the
timely execution of the critical workload under fully non-preemptive
scheduling. Notice that, as mentioned earlier, it is assumed that the schedule is resumed with the highest priority active job after restart.


\begin{theorem}
\label{thm:resp_time_nonp} 
A set of non-preemptive periodic tasks is RBR-feasible
under fixed-priority if the response time $R_i$ of each task $\tau_i$,
calculated through following relation, satisfies the condition: $\forall \tau_i
\in \mathcal{T}; R_i \le D_i$.
\begin{equation}
R_i = \max_{k\in[1,K_i]} \{F_{i,k} - (k-1)T_i\}
\end{equation}
where $F_{i,k}$ is the finishing time of job $\tau_{i,k}$ given by
\begin{equation}
\label{eq:nonprerespplusblockingtime}
F_{i,k} = S_{i,k} + C_i
\end{equation}
Here, $S_{i,k}$ is the start time of job $\tau_{i,k}$, obtained for
the smallest value that satisfies $S_{i,k}^{(q+1)} = S_{i,k}^{(q)}$ in the
following relation
\begin{equation}
S^{(k+1)}_{i,k} = B_i +  \sum_{\tau_j \in hp(\pi_i)}\left(\floor[\bigg]{\frac{S_{i,k}^{(k)}}{T_j}}+1\right) C_j +  \mathcal{O}_{i}^{np}
\label{eq:reset_resp}
\end{equation}
In Equation~\ref{eq:reset_resp}, term $B_i$ is the blocking from low
priority tasks and is calculated as $B_i = \max_{\tau_j \in
  lp{(\pi_i)}}{\{C_j\}}.$ The term $\mathcal{O}_{i}^{np}$ represents the
overhead on task execution introduced by restarts and is calculated as
follows:
\begin{equation}
\mathcal{O}_{i}^{np} = 
\left\{\begin{array}{lr}
C_r + max\left\{ \{C_j \mid j \in hp(\pi_i) \}\cup C_i \right\} & i \leq n_c\\
0 & i > n_c 
\end{array}\right.
\label{eq:nonpresetinterf}
\end{equation}
\end{theorem} 

\begin{proof}
  Equation~\ref{eq:reset_resp}
  and~\ref{eq:nonprerespplusblockingtime}, without the restart
  overhead term $\mathcal{O}_{i}^{np}$, are proposed in~\cite{4271700,
    6164261} to calculate the worst-case start time and response time
  of a task under non-preemptive setting.

  We need to show that the overhead term can be computed using
  Equation~\ref{eq:nonpresetinterf}. Under non-preemptive discipline,
  restart only impacts a single task executing on the CPU at the
  instant of restart. There are two possible scenarios that may result
  in the worst-case restart delay on finish time of task
  $\tau_i$. First, when $\tau_i$ is waiting for the higher priority
  tasks to finish their execution, a restart can occur during the
  execution of one of the higher priority tasks $\tau_j$ and delay the
  start time $\tau_i$ by $C_r+C_j$. Alternatively, a restart can occur
  infinitesimal time prior to the completion of $\tau_i$ and cause an
  overhead of $C_r + C_i$.  Hence, the worst-case delay due to a restart
  is caused by the task with the longest execution time among the
  task itself and the tasks with higher
  priority~(Equation~\ref{eq:nonpresetinterf}). The restart overhead is not included in the response-time of non-critical tasks~($\mathcal{O}_{i}^{np}=0$ for $i>n_c$).
  \end{proof}

\begin{figure}[ht]
	\centering
	\includegraphics[trim={.3cm 0.2cm 0cm 0},clip,
	width=0.8\linewidth]{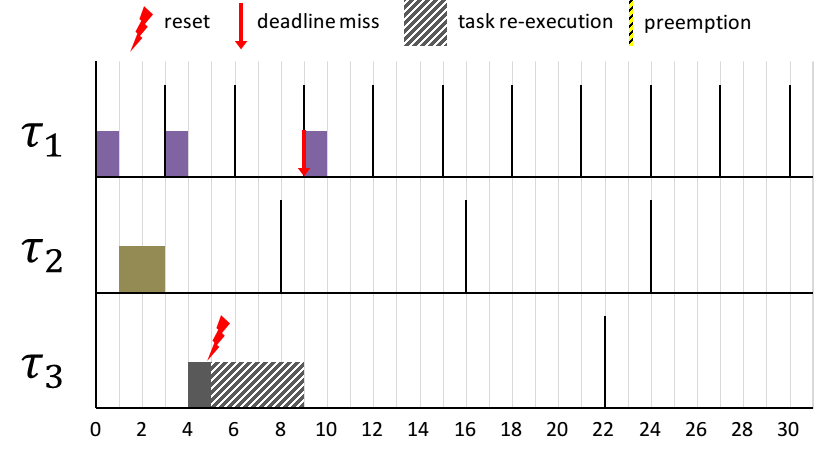}
	\caption{Example of fully non-preemptive system with 3 tasks $\tau_1
		= (1, 3); \tau_2 = (2, 8); \tau_3 = (4, 22)$, and restart at $t =
		5 - \epsilon$ ($C_r = 0$). Restart and task re-execution causes a
		deadline miss at $t = 9$.}
	\label{fig:sched_nonpreempt}
\end{figure}

Unfortunately, under non-preemptive scheduling, blocking time due to
low priority tasks, may cause higher priority tasks with short
deadlines to be non-schedulable. As a result, when preemptions are disabled, there exist task sets with arbitrary low utilization that despite having the
lowest restart overhead, are not RBR-Feasible. Figure~\ref{fig:sched_nonpreempt} uses the
same task parameters as in Figure~\ref{fig:sched_preempt}. The plot
shows that the considered task system is not schedulable under fully
non-preemptive scheduling when a restart is triggered at $t = 5 -
\epsilon$.



\section{Limited Preemptions}
\label{sec:limited-sched}

In the previous section, we analyzed the RBR-Feasibility of task sets
under fully preemptive and fully non-preemptive scheduling. Under full
preemption, restarts can cause a significant overhead because the
longest preemption chain can contains all the tasks. On the other hand,
under non-preemptive scheduling, the restart overhead is
minimum. However, due to additional blocking on higher priority
tasks, some task sets, even with low utilization, are not schedulable.

In this section we discuss two alternative models with limited
preemption. Limited preemption models are suitable for restart-based
recovery since they enable the necessary preemptions for the
schedulability of the task set, but avoid many unnecessary preemptions
that occur in fully preemptive scheduling. Consequently, they induce
lower restarting overhead and exhibit higher schedulability.

\subsection{Preemptive tasks with Non-Preemptive Ending}
\label{sec:npregions}

As seen in the previous sections, reducing the number and length of
preempted tasks in the longest preemption chain, can reduce the
overhead of restarting and increase the RBR-Feasibility of task sets. On
the other hand, preventing preemptions entirely is not desirable since
it can impact feasibility of the high priority tasks with short
deadlines. As a result, we consider a hybrid preemption model in
which, a job once executed for longer than $C_i-Q_i$ time units,
switches to non-preemptive mode and continues to execute until its
termination point. Such a model allows a job that has mostly completed
to terminate, instead of being preempted by a higher priority
task. $Q_i$ is called the size of non-preemptive ending interval of $\tau_i$
and $Q_i \leq C_i$. The model we utilize in this section, is a special
case of the model proposed in~\cite{1508455} which aims to decrease
the preemption overhead due to context switch in real-time operating
systems. In Figure~\ref{fig:sched_npreempt_reg1}, we consider a task
set with the same parameters as in Figure~\ref{fig:sched_preempt},
where in addition task $\tau_3$ has a non-preemptive region of length
$Q_3 = 1$. The preemption chain that caused the system in
Figure~\ref{fig:sched_preempt} to be non-schedulable cannot occur and
the instance of the task becomes schedulable under restarts. With the same setup,
Figure~\ref{fig:sched_npreempt_reg2} considers the case when a reset
occurs at $t = 9 - \epsilon$.

\subsubsection{RBR-Feasibility Analysis}
Theorem~\ref{theorem:withblocking} provides the RBR-feasibility
conditions of a task-set with non-preemptive ending intervals. In this
theorem, $S_{i,k}$ represents the worst case start time of the
non-preemptive region of the re-executed instance of job
$\tau_{i,k}$. Similarly, $F_{i,k}$ is used to represent the worst-case
finish time. The arrival time of instance $k$ of task $\tau_{i,k}$ is
$(k-1)T_i$.


\begin{theorem}
\label{theorem:withblocking}
A set of periodic tasks $\mathcal{T}$ with non-preemptive
ending regions of length $Q_i$, is RBR-Feasible under a fixed priority
algorithm if the worst-case response time $R_i$ of each task $\tau_i$,
calculated from Equation~\ref{eq:wcrtwithnpregions}, satisfies the
condition: $\forall \tau_i \in \mathcal{T}, R_i \le D_i$.
	\begin{equation}
	\label{eq:wcrtwithnpregions}
	R_i = \max_{k\in[1,K_i]} \{F_{i,k} - (k-1)T_i\}
	\end{equation}
	where
	\begin{equation}
	\label{eq:respplusblockingtime}
	F_{i,k} = S_{i,k} + Q_i
	\end{equation}
	and $S_{i,k}$ is obtained for the smallest value of $q$ for
        which we have $S_{i,k}^{(q+1)} = S_{i,k}^{(q)}$ in the
        following
	\begin{multline}
	\label{eq:respwithblocking}
	S_{i,k}^{(q+1)} = B_i + (k-1)C_i + C_i - Q_i\\ + \sum_{\tau_j
          \in
          hp(\tau_i)}\left(\floor[\bigg]{\frac{S_{i,k}^{(q)}}{T_j}}+1\right)
        C_j + \mathcal{O}_{i}^{npe}
	\end{multline}
        Here, the term $B_i$ is the blocking from low priority tasks
        and is calculated by
\begin{equation}
\label{eq:blockingTimeThm3}
	B_i = \max_{\tau_k \in lp{(\pi_i)}}{\{Q_k\}}.
\end{equation}
$\mathcal{O}_i^{npe}$ is the maximum overhead of the restart on the response
time and is calculated as follows:
	\begin{equation}
	\label{eq:interwithblocking}
	\mathcal{O}_i^{npe} = \left\{
	\begin{array}{lr}
C_r +  \mathcal{WCWE}(i)&i \leq n_c\\
0&i>n_c
	\end{array}\right.
	\end{equation}
	where $\mathcal{WCWE}(i)$ is the worst-case amount of the
        execution that may be wasted due to the restarts. It is given
        by the following where $\mathcal{WCWE}(1) = C_1$ and
	\begin{equation}
	\label{eq:wcwe}
	\mathcal{WCWE}(i) = C_i + max\bigg(0, \mathcal{WCWE}(i-1) - Q_{i} \bigg)
	\end{equation}
	$K_i$ in Equation~\ref{eq:wcrtwithnpregions} can be computed from Equation~\ref{eq:busyperiod} by using $\mathcal{O}_{i}^{npe}$ instead of $\mathcal{O}_{i}^{np}$.	
\end{theorem}

\begin{proof}
  Authors in~\cite{wang1999scheduling} show that the worst-case
  response time of task $\tau_i$ is the maximum difference between
  the worst case finish time and the arrival time of the jobs that arrive within
  the level-$i$ active period~(Equation~\ref{eq:wcrtwithnpregions}).

  Hence, we must compute the worst-case finish time of job $\tau_{i,k}$ in the presence of restarts. When a restart occurs
  during the execution of $\tau_{i,k}$ or while it is in preempted
  state, $\tau_{i,k}$ needs to re-execute. Therefore, the finish time
  of the $\tau_{i,k}$ is when the re-executed instance completes. As a
  result, to obtain the worst-case finish time of $\tau_{i,k}$, we
  calculate the response time of each instance when a restart with
  longest overhead has impacted that instance. We break down the
  worst-case finish time of $\tau_{i,k}$ into two intervals:  the worst-case start time of the
  non-preemptive region of the re-executed job and the length of the
  non-preemptive region,
  $Q_i$~(Equation~\ref{eq:respplusblockingtime}). $S_{i,k}$ in
  Equation~\ref{eq:respplusblockingtime}, is the worst-case start time
  of non-preemptive region of job $\tau_{i,k}$ which can be
  iteratively obtained from
  Equation~\ref{eq:respwithblocking}. Equation~\ref{eq:respwithblocking}
  is an extension of the start time computation
  from~\cite{6164261}. In the presence of non-preemptive regions, an
  additional blocking factor $B_i$ must be considered for each task
  $\tau_i$, equal to the longest non-preemptive region of the lower
  priority tasks. Therefore, the maximum blocking time that $\tau_i$
  may experience is $ B_i = \max_{\tau_j \in
    lp{(\pi_i)}}{\{Q_j\}}$. $B_i$ is added to the worst-case start
  time of the task in Equation~\ref{eq:respwithblocking}.

	
  For a task $\tau_i$ with the non-preemptive region of size $Q_i$,
  there are two cases that may lead to the worst-case wasted
  time. First case is when the system restarts immediately prior to the completion
  of $\tau_i$, in which case the wasted time is $C_i$. Second case
  occurs when $\tau_i$ is preempted immediately before the
  non-preemptive region begins~(\ie at $C_i-Q_i$) by the higher
  priority task $\tau_{i-1}$. In this case, the wasted execution is
  $C_i - Q_i$ plus the maximum amount of the execution of the higher priority tasks
  that may be wasted due to the restarts~(\ie $\mathcal{WCWE} ( i - 1)$). The
  worst-case wasted execution is the maximum of these two values \ie
  $\mathcal{WCWE}(i) = max(C_i,C_i-Q_i+\mathcal{WCWE}(i-1)) = C_i +
  max(0, \mathcal{WCWE}(i-1)-Q_i)$. Similarly, 
  $\mathcal{WCWE}(i-1)$ can be computed recursively.
\end{proof}

\begin{figure}
	\centering
	\includegraphics[trim={.3cm 0.2cm 0cm 0},clip,
	width=0.8\linewidth]{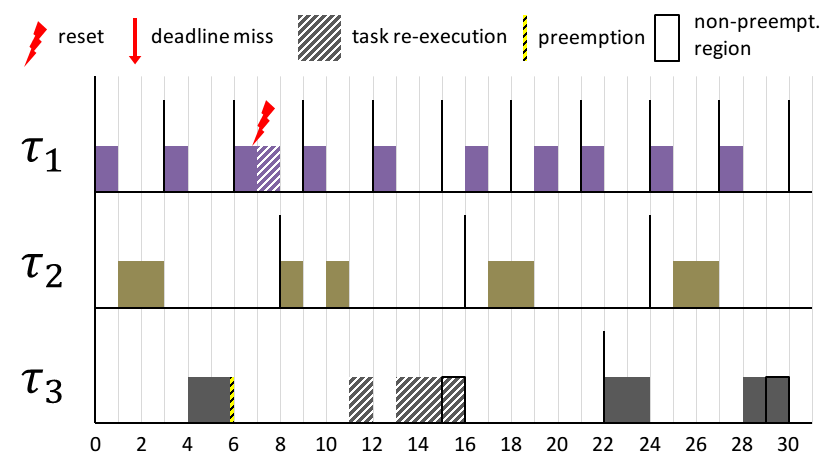}
	\caption{Example of system with 3 tasks $\tau_1 = (1, 3); \tau_2 =
		(2, 8); \tau_3 = (4, 22)$, where $\tau_3$ has a non-preemptive
		region of size $Q_3 = 1$. Restart occurs at $t = 7 - \epsilon$
		($C_r = 0$). The task set is schedulable with restarts.}
	\label{fig:sched_npreempt_reg1}
\end{figure}

\begin{figure}
	\centering
	\includegraphics[trim={.3cm 0.2cm 0cm 0},clip,
	width=0.8\linewidth]{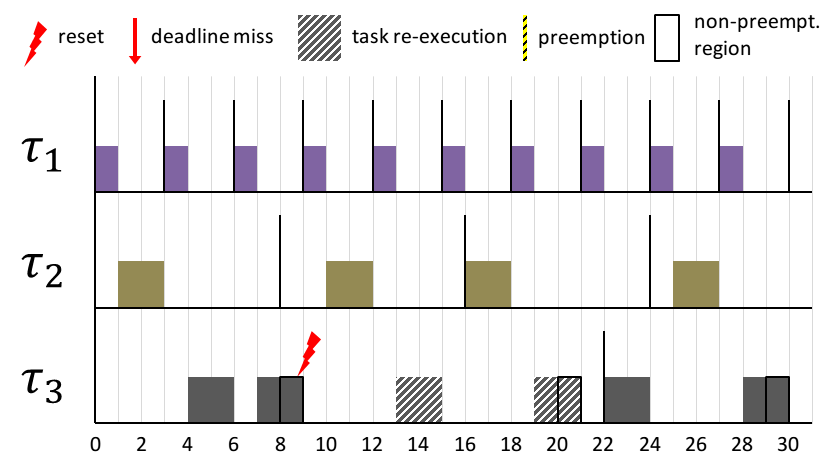}
	\caption{Example of system with 3 tasks $\tau_1 = (1, 3); \tau_2 =
		(2, 8); \tau_3 = (4, 22)$, where $\tau_3$ has a non-preemptive
		region of size $Q_3 = 1$. Restart occurs at $t = 9 - \epsilon$
		($C_r = 0$). The task set is schedulable with restarts.}
	\label{fig:sched_npreempt_reg2}
\end{figure}

\subsubsection{Optimal Size of Non-Preemptive Regions}

RBR-Feasibility of a taskset depends on the choice of $Q_i$s for the
tasks. In this section, we present an approach to determine the size
of non-preemptive regions $Q_i$ for the tasks to maximize the
RBR-Feasibility of the task set.

First, we introduce the the notion of \textit{blocking tolerance} of a
task $\beta_i$. $\beta_i$ is the maximum time units that task $\tau_i$ may be
blocked by the lower priority tasks, while it can still meet its
deadline.  Algorithm~\ref{alg:findbetta}, uses binary search and the
response time analysis of task~(from
Theorem~\ref{theorem:withblocking}) to find $\beta_i$ for a task
$\tau_i$.

\begin{algorithm}[ht]
	\SetAlgoLined\DontPrintSemicolon
	\SetKwFunction{algo}{FindBlockingTolerance}
	\SetKwProg{myalg}{}{}{}
	\myalg{\algo{$\tau_i, \mathcal{T}, Q_1,...,Q_i$ }
	}
	{	
		start = 0; end = $T_i$ /* Initialize the interval */ \\	
		\textbf{if} $R_i$(start) $> T_i$ \textbf{then} return $\tau_i$ Not Schedulable;\\
		\While{end - start $> \epsilon$}{
			middle = (start + end)/2 \\
				\textbf{if} $R_{i,B_i=\text{middle}} > T_i$ \textbf{then} end = middle ;\\
				\textbf{else} start = middle\\
		}
			return $\beta_i = $ start;
	}
	\caption{Binary Search for Finding $\beta_i$}
	\label{alg:findbetta}
\end{algorithm}

In Algorithm~\ref{alg:findbetta}, $R_{i,B_i=middle}$ is computed as
described in
Theorem~\ref{theorem:withblocking}~(Equation~\ref{eq:wcrtwithnpregions}),
where instead of using the $B_i$ from
Equation~\ref{eq:blockingTimeThm3}, the blocking time is set to the value of 
$middle$.

Note that, if Algorithm~\ref{alg:findbetta} cannot find a $\beta_i$
for task $\tau_i$, this task is not schedulable at all. This indicates that there is not any selection of $Q_i$s that would make $\mathcal{T}$ RBR-Feasible. 

Given that task $\tau_1$ has the highest priority, it may not be
preempted by any other task; hence we set $Q_1 = C_1$. The next
theorem shows how to drive optimal $Q_i$ for the rest of the tasks in
$\mathcal{T}$. The results are optimal, meaning that if there is at
least one set of $Q_i$s under which $\mathcal{T}$ is RBR-Feasible, it
will find them.

\begin{theorem}
	\label{theorem:maxblockingtimes}
	The optimal set of non-preemptive interval $Q_i$s of tasks
        $\tau_i$ for $2 \leq i \leq n$ is given by:
	\begin{equation}
	\label{eq:maxblockingtimes}
	Q_i = min\big\{min\{\beta_j \mid j\in hp(\pi_i)\},C_i\big\}
	\end{equation}
	assuming that $\beta_j \geq 0$ for $j \in hp(\pi_i)$.
\end{theorem}

\begin{proof}

  Increasing the length of $Q_i$ for a task reduces the response time
  in two ways. First, from Equation~\ref{eq:respwithblocking},
  increasing $Q_i$ reduces the start time of the job $S_{i,k}$ which
  reduces the finish time and consequently the response time of
  $\tau_i$. Second, from Equation~\ref{eq:wcwe}, increasing $Q_i$
  reduces the restart overhead $\mathcal{O}^{npe}_{i}$ on the task and
  lower priority tasks which in turn reduces the response time. Thus
  $Q_i$ may increase as much as possible up to the worst-case
  execution time $C_i$; $Q_i \leq C_i$. However, the choice of $Q_i$
  must not make any of the higher priority tasks unschedulable. As a
  result, $Q_i$ must be smaller than the smallest blocking tolerance
  of all the tasks with higher priority than $\pi_i$; $Q_i \leq
  min\{\beta_j| j\in hp(\pi_i)\}$. Combining these two conditions
  results in the relation of Equation~\ref{eq:maxblockingtimes}.
\end{proof}

\subsection{Preemption Thresholds}
\label{sec:premthresh}


In the previous section, we discussed non-preemptive endings as a way
to reduce the length of the longest preemption chain and decrease the
overhead of restarts. In this section, we discuss an alternative
approach to reduce the number of tasks in the longest preemption chain
and thus reduce the overhead of restart-based recovery.

To achieve this goal, we use the notion of preemption thresholds which
has been proposed in~\cite{wang1999scheduling}. According to this
model, each task $\tau_i$ is assigned a nominal priority $\pi_i$ and a
preemption threshold $\lambda_i \geq \pi_i$. In this case, $\tau_i$
can be preempted by $\tau_h$ only if $\pi_h > \lambda_i$. At
activation time, priority of $\tau_i$ is set to the nominal value
$\pi_i$. The nominal priority is maintained as long as the task is
kept in the ready queue. During this interval, the execution of
$\tau_i$ can be delayed by all tasks with priority $\pi_h > \pi_i$,
and by at most one lower priority task with threshold $\lambda_l \geq
\pi_i$. When all such tasks complete, $\tau_i$ is dispatched for
execution, and its priority is raised to $\lambda_i$. During
execution, $\tau_i$ can be preempted by tasks with priority $\pi_h >
\lambda_i$. When $\tau_i$ is preempted, its priority is kept at
$\lambda_i$.


%

Restarts may increase the response time of $\tau_{i,k}$ in one of two
ways; A restart may occur after the arrival of the job but
before it has started, delaying its start time
$S_{i,k}$. Alternatively, the system can be restarted after the job has
started. We use $\mathcal{O}^{pt,s}_{i}$ to denote the worst-case overhead of a restart that occurs before the start time of a job in task sets with preemption thresholds. And, $\mathcal{O}^{pt,f}_{i}$ is used to represent the worst-case overhead of a restart that occurs after the start time of a job in task sets with preemption thresholds.

In Figure~\ref{fig:sched_preempt_tr_ok}, we consider a task set with
the same parameters as in Figure~\ref{fig:sched_preempt} where in
addition $\tau_2$ and $\tau_3$ have a preemption threshold equal to
$\lambda_2 = 1$ and $\lambda_3 = 2$, respectively. This assignment is
effective to prevent a long preemption chain, and the jobs do not miss their deadline when the restart occurs at $t = 7 - \epsilon$. Notice that, the
task set is still not RBR-Feasible since if the restart occurs at $t = 9 - \epsilon$, some job will miss the deadline,
as shown in Figure~\ref{fig:sched_preempt_tr_fail}.

\begin{figure}
  \centering
  \includegraphics[trim={.3cm 0.2cm 0cm 0},clip,
  width=0.8\linewidth]{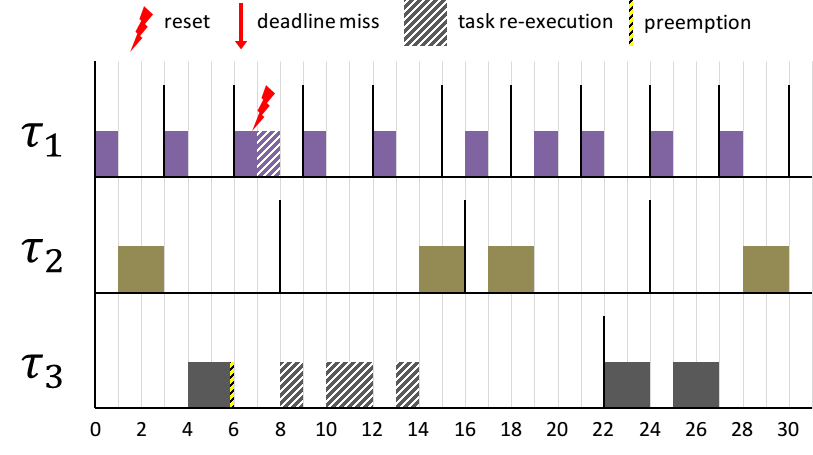}
  \caption{Example of system with 3 tasks $\tau_1 = (1, 3); \tau_2 =
    (2, 8); \tau_3 = (4, 22)$, where $\tau_2$ and $\tau_3$ have a
    preemption threshold of $\lambda_2 = 1$ and $\lambda_3 = 2$,
    respectively. Restart occurs at $t = 7 - \epsilon$ ($C_r = 0$). In
    this case, the task set remains schedulable.}
  \label{fig:sched_preempt_tr_ok}
\end{figure}

\begin{figure}
  \centering
  \includegraphics[trim={.3cm 0.2cm 0cm 0},clip,
  width=0.8\linewidth]{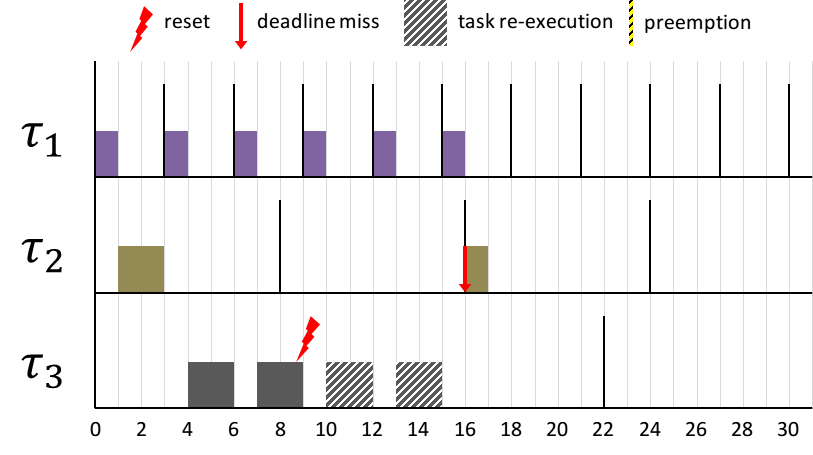}
  \caption{Example of system with 3 tasks $\tau_1 = (1, 3); \tau_2 =
    (2, 8); \tau_3 = (4, 22)$, where $\tau_2$ and $\tau_3$ have a
    preemption threshold of $\lambda_2 = 1$ and $\lambda_3 = 2$,
    respectively. Restart occurs at $t = 9 - \epsilon$ ($C_r =
    0$). The task set is not schedulable.}
  \label{fig:sched_preempt_tr_fail}
\end{figure}

\begin{theorem}
  \label{thm:thm6}
  For a task set with preemption thresholds under fixed priority, the worst-case overhead of a restart that occurs after the start of the job $\tau_{i,k}$ is $\mathcal{O}^{pt,f}_{i} = C_r +
  \mathcal{WCWE}(i)$ where
	\begin{equation}
	\label{eq:WCWE-preemption-thresholds}
	\mathcal{WCWE}(i) = C_i + \max\{\mathcal{WCWC}(j) \mid \tau_j \in hp(\lambda_i) \}
	\end{equation}
	Here, $\mathcal{WCWC}(1) = C_1$.
\end{theorem}

\begin{proof}
  After a job $\tau_{i,k}$ starts, its priority is raised to $\lambda_i$. 
  In this case, the restart will create the worst-case overhead if it occurs at the end of longest preemption chain
  that includes $\tau_i$ and any subset of the tasks with $\pi_h > \lambda_i$.
  Equation~\ref{eq:WCWE-preemption-thresholds} uses a
  recursive relation to calculate the length of longest preemption
  chain consisting of $\tau_i$ and all the tasks with $\pi_h >
  \lambda_i$.
\end{proof}

\begin{theorem}
	\label{thm:lastTheorem}
  For a task set with preemption thresholds under fixed priority, a
  restart occurring before the start time of a job $\tau_{i,k}$, can
  cause the worst-case overhead of
	\begin{equation}
          \mathcal{O}^{pt,s}_{i} = C_r +  \max\{\mathcal{WCWE}(j) \mid \tau_j \in hp(\pi_i) \}
	\label{eq:max-o-s}
	\end{equation}
	where $\mathcal{WCWE}(j)$ can be computed from
        Equation~\ref{eq:WCWE-preemption-thresholds}.
\end{theorem}
\begin{proof}
  Start time of a task can be delayed by a restart impacting any of
  the tasks with priority higher than
  $\pi_i$. Equation~\ref{eq:max-o-s} recursively finds the longest
  possible preemption chain consisting of any subset of tasks with $\pi_h > \pi_i$.
\end{proof}

Due to the assumption of one fault per hyper-period, each job may be
impacted by at most one of $\mathcal{O}^{pt,f}_{i}$ or
$\mathcal{O}^{pt,s}_{i}$, but not both at the same time. Hence, we
compute the finish time of the task once assuming that the restart
occurs before the start time \ie $\mathcal{O}^{pt,f}_{i} = 0$, and
another time assuming it occurs after the start time \ie
$\mathcal{O}^{pt,s}_{i} = 0$. Finish time in these two cases is
referred respectively by $F^s_{i,k}$~(restart before the start time)
and $F^f_{i,k}$~(restart after the start time).

We expand the response time analysis of tasks with preemption
thresholds from~\cite{wang1999scheduling}, considering the overhead of
restarting. In the following, $S_{i,k}$ and $F_{i,k}$ represent the
worst case start time and finish time of job
$\tau_{i,k}$. And, the arrival time of $\tau_{i,k}$ is $(k-1)T_i$. The
worst-case response time of task $\tau_i$ is given by:
\begin{equation}
R_i = \max_{k\in[1,K_i]} \bigg\{\max\{F^s_{i,k},F^f_{i,k}\} - (k-1)T_i\bigg\}
\end{equation}
Here, $K_i$ can be obtained from Equation~\ref{eq:busyperiod} by using $max(\mathcal{O}^{pt,f}_{i},\mathcal{O}^{pt,s}_{i})$ instead of $\mathcal{O}^{np}_{i}$. A task
$\tau_i$ can be blocked only by lower priority tasks that cannot be
preempted by it, that is:
\begin{equation}
B_i = \max_j\{C_j \mid \pi_j < \pi_i \leq \lambda_j\}
\end{equation}
To compute finish time, $S_{i,k}$ is computed iteratively using the
following equation~\cite{wang1999scheduling}:
\begin{equation}
\label{eq:starttime}
S_{i,k}^{(q)} = B_i + (k-1)C_i +\\ \sum_{j\in hp(\pi_i)}\bigg(1+ \bigg\lfloor \frac{S_{i,k}^{(q-1)}}{T_j} \bigg\rfloor \bigg)C_j + \mathcal{O}^{pt,s}_{i}
\end{equation}
Once the job starts executing, only the tasks with higher priority
than $\lambda_i$ can preempt it. Hence, the $F_{i,k}$ can be derived
from the following:
\begin{multline}
\label{eq:finishtime}
F_{i,k}^{(q)} = S_{i,k} + C_i + \\ \sum_{j \in hp(\lambda_i)} \left(
  \bigg\lceil \frac{F^{(q-1)}_{i,k}}{T_j} \bigg\rceil - \left( 1+
    \bigg\lfloor \frac{S_{i,k}}{T_j} \bigg\rfloor \right) \right)C_j +
\mathcal{O}^{pt,f}_{i}
\end{multline}

Task set $\mathcal{T}$ is considered RBR-Feasible if $\forall \tau_i
\in \mathcal{T}, R_i \leq T_i$.

RBR-Feasibility of a task set depends on the choice of $\lambda_i$s
for the tasks. In this paper, we use a genetic algorithm to find a set
of preemption thresholds to achieve RBR-Feasibility of the
task-set. Although this algorithm can be further improved to find the
optimal threshold assignments, the proposed genetic algorithm achieves
acceptable performance, as we show in Section~\ref{sec:evaluation}.

%
%

%
\section{Evaluation}
\label{sec:evaluation}

In this section, we compare and evaluate the four fault-tolerant
scheduling strategies discussed in this paper. In order to evaluate
the practical feasibility of our approach, we have also performed a
preliminary proof-of-concept implementation on commercial hardware~(i.MX7D platform) for an actual 3 degree-of-freedom helicopter. We tested logical faults, application faults and system-level faults and demonstrated that the physical system remained within the admissible region. 
Due to space
constraints, we omit the description and evaluation of our
implementation and refer to~\cite{fault-tolerant-reset} for additional
details.

\begin{figure}[h!]
	\centering 
	\subfigure[Fully preemptive]{\label{fig::OR11}\includegraphics[width=0.241\textwidth]{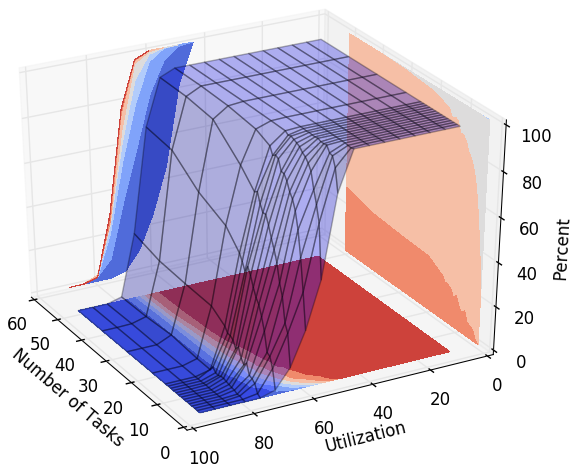}}
	\subfigure[Fully non-preemptive]{\label{fig::OR12}\includegraphics[width=0.241\textwidth]{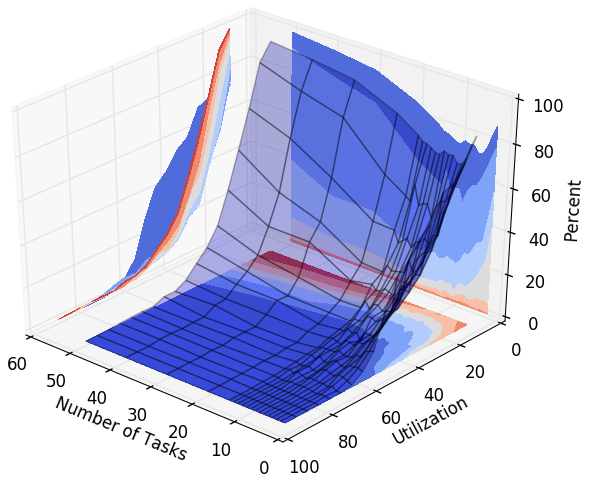}}
	\subfigure[Non-preemptive ending intervals.]{\label{fig::OR13}\includegraphics[width=0.241\textwidth]{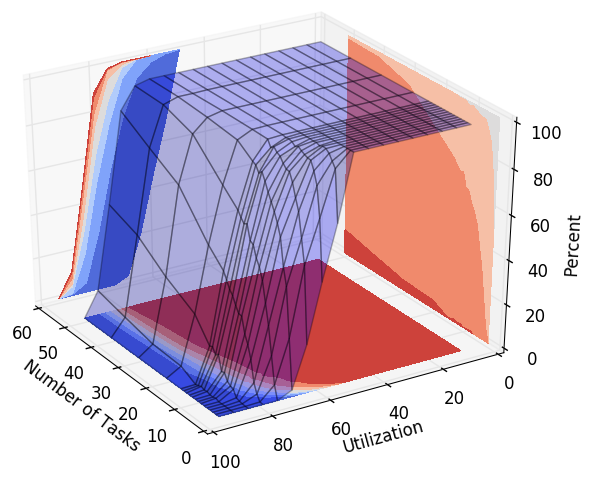}}
	\subfigure[Preemption thresholds.]{\label{fig::OR14}\includegraphics[width=0.241\textwidth]{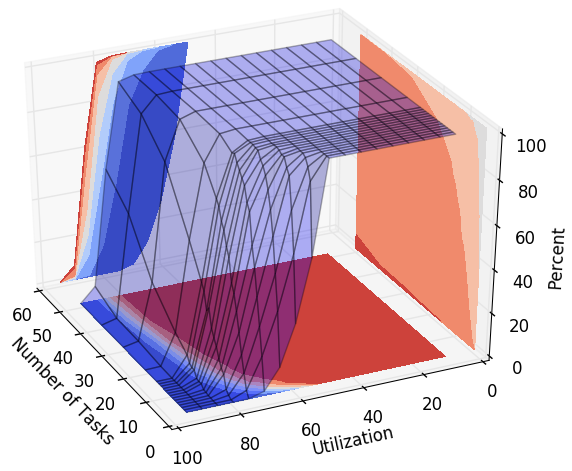}}	
	\caption{Minimum Period: 10, Maximum Period: 1000}
	\label{fig:exp90-1000}
\end{figure}

\begin{figure}[h!]
	\centering 
	\subfigure[Fully preemptive]{\label{fig::OR21}\includegraphics[width=0.241\textwidth]{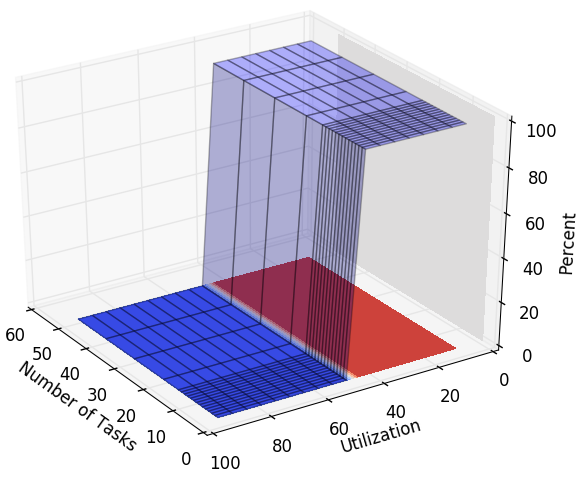}}
	\subfigure[Fully non-preemptive]{\label{fig::OR22}\includegraphics[width=0.241\textwidth]{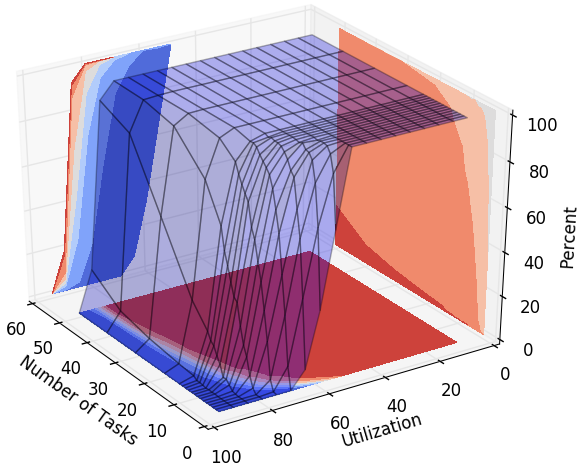}}
	\subfigure[Non-preemptive ending intervals.]{\label{fig::OR23}\includegraphics[width=0.241\textwidth]{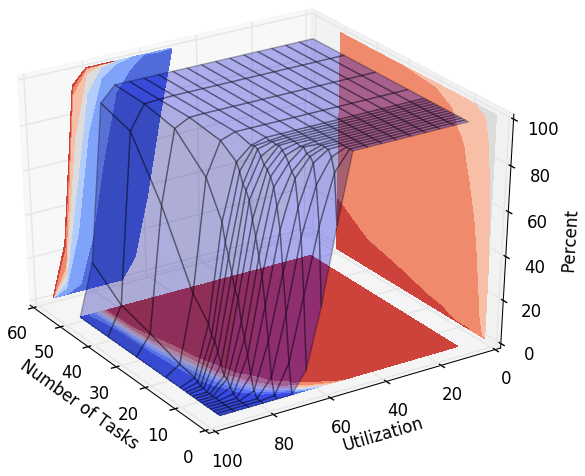}}
	\subfigure[Preemption thresholds.]{\label{fig::OR24}\includegraphics[width=0.241\textwidth]{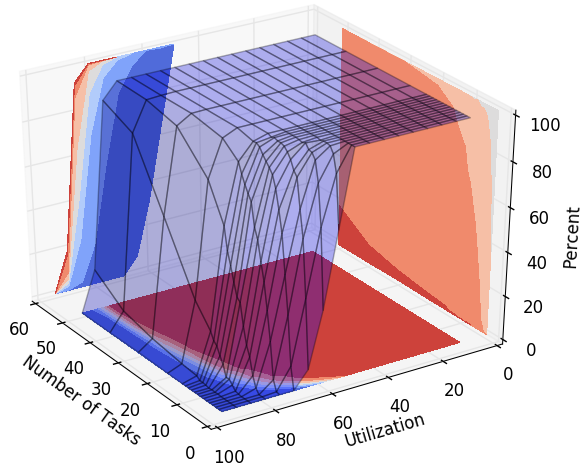}}	
	\caption{Minimum Period: 900, Maximum Period: 1000}
	\label{fig:exp900-1000}
\end{figure}

\subsection{Evaluating Performance of Scheduling Schemes}

In this section, we evaluate the performance of four fault-tolerant
scheduling schemes that are discussed in this paper. For each data
point in the experiments, 500 task sets with the specified utilization
and number of tasks are generated. Then, RBR-feasibility of the task
sets are evaluated under four discussed schemes; fully preemptive,
fully non-preemptive, non-preemptive ending intervals, and preemption
thresholds. In order to evaluate performance of the scheduling
schemes, all the tasks in the analysis are assumed to be part of the
critical workload. Priorities of the tasks are assigned according to
the periods, so a task with shorter period has a higher priority.

The experiments are performed with two sets of parameters for the
periods of the task sets. In the first set of
experiments~(Figure~\ref{fig:exp90-1000}), task sets are generated
with periods in the range of 10 to 1000 time units. In the second
set~(Figure~\ref{fig:exp900-1000}), tasks have a period in the range
of 900 to 1000 time units. As a result, tasks in the first experiment
have more diverse set of periods than the second one.

As shown in Figure~\ref{fig::OR11} and~\ref{fig::OR21}, all the task
sets with utilization less than $50\%$ are RBR-feasible under
preemptive scheduling. This observation is consistent with the results
of~\cite{729793} which considers preemptive task sets under rate
monotonic scheduling with a recovery strategy similar to
ours~(re-executing all the unfinished tasks), and shows that all the
task sets with utilization under $50\%$ are schedulable.

Moreover, a comparison between Figure~\ref{fig::OR11}
and~\ref{fig::OR21} reveals that fully preemptive setting performs
better when tasks in the task set have diverse rates. To understand
this effect, we must notice that the longest preemption chain for a
task in preemptive setting, consists of the execution time of all the
tasks with a higher priority. Therefore, under this scheduling
strategy, tasks with low priority are the bottleneck for
RBR-feasibility analysis. When the diversity of the periods is
increased, lower priority tasks, on average, have much longer
periods. As a result, they have a larger slack to tolerate the
overhead of restarts compared to the lower priority tasks in task sets
with less diverse periods. Hence, more task sets are RBR-feasible when
a larger range of periods is considered.

On the contrary, when tasks have more diverse periods, non-preemptive
setting performs worse~(Figure~\ref{fig::OR12} and
\ref{fig::OR22}). This is because, with diverse periods, tasks with
shorter periods~(and higher priorities) experience longer blocking
times due to low priority tasks with long execution times.

As the figures show, scheduling with preemption thresholds and
non-preemptive intervals in both experiments yield better performance
than preemptive and non-preemptive schemes. This effect is expected
because the flexibility of these schemes allows them to decrease the
overhead of restarts by increasing the non-preemptive regions, or by
increasing the preemption thresholds while maintaining the feasibility
of the task sets. Tasks under these disciplines exhibit less blocking
and lower restart overhead.

Preemption thresholds and non-preemptive endings in general
demonstrate comparable performance. However, in task sets with very
small number of tasks~(2-10 task), scheduling using non-preemptive
ending intervals performs slightly better than preemption
thresholds. This is due to the fact that, with small number of tasks,
the granularity of the latter approach is limited because few choices
can be made on the tasks' preemption thresholds. Whereas, the length
of non-preemptive intervals can be selected with a finer granularity
and is not impacted by the number of tasks.

\section{Related Work}
\label{sec:related}



Most of the previous work on Simplex
Architecture~\cite{sha1998dependable, seto1999engineering, Sha01usingsimplicity, sha1996evolving,crenshaw2007simplex} has
focused on design of the switching logic of DM or the SC, assuming
that the underlying RTOS, libraries and middle-ware will correctly
execute the SC and DM. Often however, these underlying software layers
are unverified and may contain bugs. Unfortunately, Simplex-based
systems are not guaranteed to behave correctly in presence of
system-level faults.
System-Level Simplex and its variants~\cite{bak2009system,
  fardin2016reset, mohan2013s3a} run SC and DM as bare-metal
applications on an isolated, dedicated hardware unit. By doing so, the
critical components are protected from the faults in the OS or
middle-ware of the complex subsystem. However, exercising this design
on most multi-core platforms is challenging. The majority of
commercial multi-core platforms are not designed to achieve strong
inter-core fault isolation due to the high-degree of hardware resource
sharing. For instance, a fault occurring in a core with the highest
privilege level may compromise power and clock configuration of the
entire platform. To achieve full isolation and independence, one has
to utilize two separate boards/systems. Our design enables the system to safely tolerate and recover from
application-level and system-level faults that cause silent failures
in SC and DM {\bf without utilizing additional hardware}.



The notion of restarting as a means of recovery from faults and
improving system availability was previously studied in the
literature. Most of the previous work, however, target traditional
\textit{non}-safety-critical computing systems such as servers and
switches. Authors in~\cite{candea2001recursive} introduce recursively
restartable systems as a design paradigm for highly available systems. Earlier
literature~\cite{Candea03crash-onlysoftware, 
  candea2004microreboot} illustrates the concept of micro-reboot which
consists of having fine-grain rebootable components and trying to
restart them from the smallest component to the biggest one in the
presence of faults. The works in~\cite{vaidyanathan2005comprehensive,
  garg1995analysis, huang1995software} focus on failure and fault
modeling and try to find an optimal rejuvenation strategy for various non safety-critical
systems.

In the context of safety-critical CPS, authors in~\cite{abdi2017application} propose the procedures to design a base controller that enables the entire computing system to be safely restarted at run-time. Base Controller keeps the system inside a subset of safety region by updating the actuator input at least once after every system restart.
In~\cite{fardin2016reset}, which is variation of System-Level Simplex, authors propose that the
complex subsystem can be restarted upon the occurrence of
faults. In this design, safe restarting is possible because the back up controller runs on a
dedicated processing unit and is not impacted by the restarts in the
complex subsystem.


One way to achieve fault-tolerance in real-time systems is to use time
redundancy. Using time redundancy, whenever a fault leads to an error,
and the error is detected, the faulty task is either re-executed or a
different logic (recovery block) is executed to recover from the
error. It is necessary that such recovery strategy does not cause any
deadline misses in the task set. Fault tolerant scheduling has been
extensively studied in the literature. Hereby we briefly survey those
works that are more closely related. 
A feasibility check algorithm
under multiple faults, assuming EDF scheduling for aperiodic preemptive tasks is proposed in~\cite{869322}. 
An exact schedulability tests using checkpointing
for task sets under fully preemptive model and transient fault that
affects one task is proposed in~\cite{Punnekkat2001}. This analysis is
further extended in~\cite{Lima2005} for the case of multiple faults as
well as for the case where the priority of a critical task’s recovery
block is increased.
In \cite{5591653}, authors propose the exact feasibility test for
fixed-priority scheduling of a periodic task set to tolerate multiple
transient faults on uniprocessor. In~\cite{1183950} an approach is
presented to schedule under fixed priority-driven preemptive
scheduling at least one of the two versions of the task; simple
version with reliable timing or complex version with potentially
faulty.
Authors in~\cite{729793} consider a similar fault model to ours, where
the recovery action is to re-execute all the partially executed tasks
at the instant of the fault detection \ie executing task and all the
preempted tasks. This work only considers preemptive task sets under
rate monotonic and shows that single faults with a minimum
inter-arrival time of largest period in the task set can be recovered
if the processor utilization is less than or equal to $50\%$.
In \cite{6925991}, the authors investigate the feasibility of task
sets under fault bursts with preemptive scheduling. Similar to our work, the recovery action is
to re-execute the faulty job along with all the partially completed
(preempted) jobs at the time of fault detection. 
Most of these works are only applicable to transient faults~(\eg
faults that occur due to radiation or short-lived HW malfunctions)
that impact the task and do not consider faults affecting the
underlying system. Additionally, most of these works assume that an online
fault detection or acceptance test mechanism exists. While this assumption is
valid for detecting transient faults or timing faults, detecting complex
system-level faults or logical faults is non-trivial. 
Additionally, to the best of our knowledge, our paper is the first one
to provide the sufficient feasibility condition in the presence of
faults under the \emph{preemption threshold} model and task sets with
\emph{non-preemptive ending intervals}.

\section{Conclusion}
\label{sec:conclusion-future}

Restarting is considered a reliable way to recover traditional
computing systems from complex software faults. However, restarting
safety-critical CPS is challenging. In this work we propose a
restart-based fault-tolerance approach and analyze feasibility
conditions under various schedulability schemes. We analyze the
performance of these strategies for various task sets. 
This approach enables us to provide
formal safety guarantees in the presence of software faults in the
application-layer as well as system-layer faults utilizing only one
commercial off-the-shelf processor.


%
%



\bibliographystyle{IEEEtran}
{\footnotesize
\bibliography{fardin}  
}

%

\end{document}